\documentclass[12pt]{article}
\usepackage{subfigure}
\usepackage{graphicx}
\usepackage[greek,english]{babel}
\usepackage[utf8x]{inputenc}
\usepackage{amssymb,amsmath,amsthm,mathrsfs,bm}
\usepackage{hyperref}
\usepackage{color}
\usepackage{cite}
\newtheorem{theorem}{Theorem}

\newtheorem{prop}{Proposition}
\numberwithin{equation}{section}
\newcommand{\newc}{\newcommand}
\newc{\ra}{\rightarrow}
\newc{\lra}{\leftrightarrow}
\newc{\be}{\begin{equation}}
\newc{\ee}{\end{equation}}
\newc{\bg}{\begin{gathered}}
\newc{\eg}{\end{gathered}}
\newc{\bs}{\begin{split}}
\newc{\es}{\end{split}}
\newc{\ba}{\begin{eqnarray}}
\newc{\ea}{\end{eqnarray}}
\newc{\ov}{\overline}
\newc{\pa}{\partial}
\newc{\D}{\Delta}

\begin{document}
\begin{titlepage}
\begin{flushright}
v.\today
\end{flushright}
\vspace{0.5cm}
\begin{center}
\begin{Large}
{\bf

Exponential mixing of all orders\\  for  Arnol'd cat map lattices 

}

\end{Large}

\hskip1.0truecm

Minos Axenides$^{(a)}$\footnote{E-Mail: axenides@inp.demokritos.gr}, Emmanuel Floratos$^{(a,b,c)}$\footnote{E-Mail: mflorato@phys.uoa.gr} and Stam Nicolis$^{(d)}$\footnote{E-Mail: stamatios.nicolis@univ-tours.fr}

\hskip1.0truecm
{\sl 
${ }^{(a)}$Institute for Nuclear and Particle Physics, NCSR ``Demokritos''\\Aghia Paraskevi, Greece\\
${ }^{(b)}$Physics Department, University of Athens\\ Athens, Greece\\
${ }^{(c)}$Research Office of Mathematical Physics and Quantum Information,\\Academy of Athens, Division of Natural Sciences, Athens 10679, Greece\\
${ }^{(d)}$Institut Denis Poisson, Université de Tours, Université d'Orléans, CNRS (UMR7013)\\
Parc Grandmont, 37200 Tours, France

}
\end{center}
\begin{abstract}
Arnol'd cat maps can describe accelerated probes of the near horizon geometry of extremal black hole spacetimes; and coupled Arnol'd cat maps can describe multiparticle probes, as well as provide a framework for the near horizon geometry itself, when the black hole microstates can be resolved. Coupled Arnol'd cat maps define lattice field theories that have the property of being intrinsically chaotic, therefore can capture salient properties of information processing by black holes. One such property is that of scrambling, which, in the classical limit, becomes mixing. So it is of interest to compute the mixing times of the corresponding field theories.
In this paper we show that the recently introduced classical Arnol'd cat map lattice field theories  are  exponentially mixing  to all orders. Their  mixing times are well-defined and  are  expressed in terms of the Lyapunov exponents, more precisely by the  combination that defines the inverse  of the  Kolmogorov-Sinai entropy of these systems. 
We prove by an explicit recursive construction of  the correlation functions, that these exhibit $l-$fold mixing for any $l= 3,4,5,\ldots$.
This computation is  relevant for Rokhlin's conjecture,  which states that  2-fold mixing induces $l-$fold mixing for any $l>2$.  Our results show  that  2-fold exponential mixing, while being necessary for any $l-$fold mixing to hold,  is nevertheless not sufficient for Arnol'd cat map lattice field theories. That $l-$fold mixing does hold is because an additional contribution to the correlation function can be shown to vanish in the long time limit.

\end{abstract}
\end{titlepage}
\tableofcontents
\section{Introduction and Motivation}\label{intro}
Controversial issues of locality and unitarity between quantum theory and black hole physics have given rise to new  concepts and principles in quantum gravity such as those of the holographic and complementarity principles and AdS/CFT correspondence, with 
black holes~\cite{Bousso:2022ntt} having been conjectured to be the fastest information scramblers in nature~\cite{Hayden:2007cs,Sekino:2008he}. 
In essence they highlight  that the no-cloning theorem of quantum information imposes constraints on the speed of information spreading on the black hole horizon
 due to the interaction of infalling matter with the  horizon degrees of freedom. Its characteristic timescale  has been estimated to be proportional to the logarithm of the entropy of the black hole. It is called the scrambling time, a conjectured lower bound to the time of information propagation in the universe, with the search for a microscopic many-body system that can saturate it being an area of topical research. 

 More specifically the ingredients of regular local field theories have been shown to be inadequate to accommodate fast information processing among the near horizon black hole microscopic degrees of freedom; this has led to the study of how nonlocality can be taken into account~\cite{Hayden:2007cs,Almheiri:2012rt,harlow2022tf1}. A further property, that  has 
 emerged as relevant is that of chaoticity: 
 
 The large$-N$ limit, that underlies the holographic correspondence, implies that the near horizon degrees of freedom, that are macroscopic in number, display molecular chaos, in the sense of Boltzmann--but, now, the fluctuations are due to quantum, rather than thermal, effects, as has been highlighted by the so-called Maldacena--Shenker-Stanford bound on the greatest Lyapunov exponent~\cite{maldacena2016bound}.
This has, in turn, motivated the study of how chaos appears in matrix models~\cite{Iizuka:2008eb}, conformal field theories 
and a large class of SYK type of models~\cite{PhysRevD.94.106002} which have been  dubbed chaotic field theories~\cite{Cvitanovic:2000kj}. 
 SYK models have been shown to provide the appropriate holographic dual description to Jackiw-Teitelboim gravity~\cite{Trunin:2020vwy,Mertens:2022irh}, though much remains to be understood about their properties.

What has received much less attention is the fact that while  black hole spacetimes, when probed by classical matter, have infinite entropy, they have {\em finite} entropy, when probed by quantum matter. Furthermore, that the fundamental problem for understanding quantum gravity is that classical gravity is  a gauge theory with a {\em noncompact} gauge group--that of diffeomorphisms; it is the noncompactness of the gauge group that leads to the inevitability of the appearance of spacetime singularities that a quantum description must resolve. How this can be achieved is not, yet, known. For the moment the only hints we have for quantum features of gravity are that superstring theory does possess the degrees of freedom that can account for the finite entropy of extremal (and near-extremal) black holes,  whose singularities are time-like and, thus, avoidable, when these are probed by quantum matter; we do not control, however, the corresponding properties of non-extremal black holes, whose singularities are space-like. 
Nor is the dynamics of the degrees of freedom, themselves, that account for the entropy of extremal black holes known in full detail. 

It is in this context that our contribution to the subject  lies.

On the one hand, we have leveraged the fact that extremal black holes have  finite entropy--and that the entropy does not change with time, since extremal black holes do not emit Hawkinhg radiation--to propose a so-called ``modular discretization''~\cite{Axenides:2013iwa} of the near horizon geometry of such black holes, when the black hole microstates can be resolved. This discretization provides a compactification of the gauge group. We have shown that this discretization can realize the AdS$_2$/CFT$_1$ correspondence as a duality for the single-particle probes of the near-horizon geometry. In ref.~\cite{Axenides:2015aha,Axenides:2016nmf} we showed that this discretization passes some quite non-trivial consistency checks, namely that the dynamics of the single-particle evolution operators satisfies the assumptions of the Eigenstate Thermalization Hypothesis, that has emerged as a key property of the dynamics of the microstates--and their probes. We have also shown that the dynamics is chaotic in a way consistent with the saturation of the ``scrambling time bound''~\cite{Hayden:2007cs,Sekino:2008he,Barbon:2013goa}; though for single-particle probes this time is, in fact,  the mixing time. 

Furthermore, we have shown that our discretization passes another non-trivial consistency check, namely that it does allow the recovery of the smooth geometry of AdS$_2$ in an appropriate scaling limit~\cite{Axenides:2019lea,Axenides:2022cwy}. 

The next step involves using this discretization for constructing, either multiparticle probes of the near-horizon geometry, or for describing the near-horizon geometry itself, in order to then understand its possible dynamics. 

This means starting by constructing lattice field theories--as a prelude to studying their scaling limits.  In this context the work of refs.~\cite{Basteiro:2022zur,Basteiro:2022pyp,erdmenger2024discrete} comes closest to our approach, which was pursued in ref.~\cite{axenides2023arnol}. In that paper we describe how to obtain the consistent equations of motion for coupled Arnol'd cat maps. One reason this is of interest is that it represents the first construction of a many-body system that does not possess an integrable limit, since each  Arnol'd cat map is a chaotic system. Another reason is that everything is under analytical control and the system is invariant under symplectic transformations (this is a main difference with the work pertaining to coupled map lattices~\cite{bunimovich2005coupled,bunimovich1988spacetime,gutkin2021linear}, that do not focus on the symmetries of the models).
The reason the Arnol'd cat map is of interest is because, on the one hand, it does represent a consistent observer of the near-horizon geometry of an extremal black hole, in particular an accelerating observer, on the other hand it captures expected  chaotic properties of the near-horizon dynamics. 

 In the present work  we pursue  the study of the chaotic properties of the many-body system, by elucidating its mixing properties, the classical avatars of scrambling.  We will show that it exhibits fast exponential mixing of all orders in the number of observables. 
To this end we adapt  an operator method developed by de Bièvre~\cite{de1995chaos}, firstly for the case of two observable functions in the computation of the decay rates of their correlation function. Moreover we generalize his method for any number of observables for the Arnol'd cat map lattice, deducing  that it is a strongly $l-$fold mixing system, for any $l=2,3,\ldots$ and that the mixng time is $1/S_\mathrm{K-S},$ where $S_\mathrm{K-S}$ is the Kolmogorov-Sinai entropy.  

The role of deterministic chaos~\cite{zaslavsky1985chaos} in rendering mixing phenomena fast and ``efficient''  is well recognized and can be understood within  the so-called ergodic hierarchy classification of hyperbolic dynamical systems~\cite{berkovitz2006ergodic}. Mixing is the classical analog of the fast scrambling of information, which is the rapid spreading--in phase space--of an initially  localized perturbation, as it evolves towards a homogeneous stationary state.  
 
Our results pertain, also,  to the well known Rokhlin conjecture~\cite{rokhlin1949endomorphisms,kosenko2019rokhlin,ryzhikov2024multiplemixing75years} of whether 2-fold mixing implies  multiple $l-$fold mixing and more generally mixing of all higher orders. We show, by explicit calculation, in our system,  that 2-fold mixing is necessary but not sufficient for any $l-$fold ($l>2$) mixing to hold. 
We proceed, to this end,  in steps, by establishing, first, both  the 2-fold and $l-$fold  mixing property of a single ACM. Subsequently we generalize this result  to   the case of the extended $n-$body coupled Arnold cat map systems through the explicit computation of the  $l-$fold correlation functions, for $l=2$ and higher.  
We find, by explicit calculation, that the correlation functions show exponential decay behavior and identify the mixing time as  the inverse of the Kolmogorov-Sinai entropy. The reason that 2-fold mixing isn't sufficient for ensuring $l-$fold mixing can be understood from the recursive construction of the $l+1$st correlation function,  which is the sum of two terms, the first being the $l-$th correlation function, while the second is a remainder term that can be shown, also, to tend to zero for large times--with subleading behavior.  

The plan of the paper is the following:

In section~\ref{symplfibseq} we review the salient features of our previous paper~\cite{axenides2023arnol}, namely the construction of the evolution operator for $n$ maps, the analytical computation of the spectrum of the Lyapunov exponents, from which we have obtained  the Kolmogorov--Sinai entropy, $S_\mathrm{K-S}.$

 In sections ~\ref{ACM1mix} ~\ref{multiplemix} we discuss the definition of  simple and multiple mixing  of chaotic dynamical systems and we provide an explicit calculation of the corresponding mixing times for the case of a single Arnol'd cat map. 

In section~\ref{CACM} we  extend the previous calculation to the case of $n,$ symplectically interacting, Arnol'd cat maps, for any $n.$ We find that the corresponding, simple, mixing time is equal to $1/S_\mathrm{K-S}.$

Our conclusions and discussion for directions of further inquiry are set forth in section~\ref{concl}.

\section{Interacting Arnol'd cat maps from symplectic  couplings of $n$, $k-$Fibonacci sequences}\label{symplfibseq} 
In this section we review the salient features of our previous paper~\cite{axenides2023arnol}, pertaining to  the construction of the evolution operator for $n$ interacting Arnol'd cat maps and the analytical computation of the spectrum of the Lyapunov exponents, from which we have obtained  the Kolmogorov--Sinai entropy, $S_\mathrm{K-S}.$

We consider a dynamical system of $n$ degrees of freedom, whose space of states is the torus $\mathbb{T}^{2n},$ the $2n-$dimensional torus of radius 1. Therefore any state can be identified with a point $\bm{x}\in\mathbb{T}^{2n}.$

We shall describe the time evolution of this system by successive applications on a point $\bm{x}\in\mathbb{T}^{2n}$ of elements, ${\sf M},$  of the symplectic group over the integers, 
Sp$_{2n}[\mathbb{Z}].$  The evolution is, thus, discrete in time. In equations,
\begin{equation}
\label{Mevol0}
\bm{x}_{m+1}=\bm{x}_m\,{\sf M}\,\mathrm{mod}\,1
\end{equation}
where $\bm{x}_m=(\bm{q}_m,\bm{p}_m)$ describes the state of the system at the $m-$th time tick. Here the positions and momenta take values in $\mathbb{T}^n.$
These equations describe the evolution of $n$ ``particles'', from the tick $m$ to the tick $m+1.$

The symplectic group, $\mathrm{Sp}_{2n}[\mathbb{Z}],$ is defined as the set of integer-valued matrices ${\sf M}$ that satisfy the relation
\begin{equation}
\label{sympltructM}
{\sf M}^{\mathrm{T}}{\sf J}{\sf M}={\sf J}
\end{equation}
where ${\sf J}$ is given by 
\begin{equation}
\label{symplstructJ}
{\sf J}=\left(\begin{array}{cc} 0 & -I\\ I & 0\end{array}\right)
\end{equation} 
Eq.~(\ref{sympltructM}) implies that, for any two vectors, $\bm{x}$ and $\bm{y}\in\mathbb{T}^{2n},$ their {\em symplectic product}
\begin{equation}
\label{symplprod}
\langle\bm{x},\bm{y}\rangle\equiv \bm{x}^\mathrm{T}{\sf J}\bm{y}
\end{equation} 
is preserved by the group action, viz.
\begin{equation}
\label{symplprodMevol}
\langle{\sf M}\bm{x},{\sf M}\bm{y}\rangle = \langle\bm{x},\bm{y}\rangle
\end{equation}
Upon decomposing ${\sf M}$ into four blocks of $n\times n$ matrices 
\begin{equation}
\label{Mblocks}
{\sf M}=\left( 
\begin{array}{cc} 
{\sf A} & {\sf B} \\ {\sf C} & {\sf D}
\end{array}
\right)
\end{equation}
we obtain the constraints on the blocks,
\begin{equation}
\label{symplconstraints}
\begin{array}{l}
{\sf A}^\mathrm{T}{\sf D}-{\sf C}^\mathrm{T}{\sf B} = I_{n\times n}\\
{\sf A}^\mathrm{T}{\sf C}={\sf C}^\mathrm{T}{\sf A}\\
{\sf B}^\mathrm{T}{\sf D}={\sf D}^\mathrm{T}{\sf B}
\end{array}
\end{equation}
It is known that Sp$_{2n}[\mathbb{Z}]$ is generated by the elements
\begin{equation}
\label{Spgen}
\begin{array}{ccc}
\displaystyle
{\sf M}_{\mathrm{R}}({\sf R})=\left(\begin{array}{cc} I & {\sf R}\\ 0 & I\end{array}\right), & 
\displaystyle
{\sf M}_{\mathrm{L}}({\sf L})=\left(\begin{array}{cc} I & 0\\ {\sf L} & I\end{array}\right) & \mathrm{and}\hskip0.3truecm
\displaystyle
{\sf D}_{\mathrm{S}}=\left(\begin{array}{cc} {\sf S}^{\mathrm{T}}& 0\\ 0 & {\sf S}^{-1}\end{array}\right)
\end{array}
\end{equation}
where ${\sf R}$ and ${\sf L}$ are symmetric matrices; all matrices have integer entries. 

In what follows we are interested in those evolution matrices, ${\sf M},$ that have strictly positive eigenvalues--the reason is that this property is sufficient for the dynamics to be chaotic. Such symplectic matrices are called {\em hyperbolic} and the corresponding dynamical systems are called {\em hyperbolic} or {\em Anosov}. 

It can be shown that the symplectic property implies that the eigenvalues come in pairs, $(\rho_+^{(i)},\rho_-^{(i)}=1/\rho_+^{(i)}<1).$ For hyperbolic matrices this implies the existence of $n$ planes, spanned by the corresponding eigenvectors: The dynamics is {\em expanding} along the eigenvectors, $\bm{v}_+^{(i)},$ corresponding to the eigenvalues  $\rho_+^{(i)}$  and {\em contracting} along the eigenvectors, $\bm{v}_-^{(i)},$ corresponding to the eigenvalues $\rho_-^{(i)}.$

We shall focus here on evolution matrices  ${\sf M}$ that are hyperbolic and symmetric, so the eigenvectors $\bm{v}_\pm^{(i)}$ span orthogonal eigenspaces and 
$\bm{v}_+^{(i)}$  and $\bm{v}_-^{(i)}$ are orthogonal for each $i=1,2,\ldots,n.$ 

This class of matrices includes, in particular, the generalization of the so-called Arnol'd cat map, defined by the $2\times 2$ matrix
\begin{equation}
\label{ACMevol}
{\sf M}=\left(\begin{array}{cc} 1 & 1 \\ 1 & 2\end{array}\right)
\end{equation}
to the case of $n$ such maps, in interaction. 

Moreover, since the phase space is compact, this class of maps exhibits strong mixing as well as $l-$fold mixing of {\em any} order, $l,$ as we shall show in the following.
To this end we shall use the Lyapunov exponents and the Kolmogorov--Sinai entropy, that were calculated in closed form in~\cite{axenides2023arnol}.

In order to study the effects of coupling of $n$ Arnol'd cat maps, each of which is  defined on a   lattice of $n$ sites, we associate to each site  a two--dimensional torus, with dynamics described by a single Arnol'd cat map.  The total phase space of the system will be $\mathbb{T}^{2n}$, the $2n-$dimensional torus and the proposed dynamics will be described by appropriate elements of the symplectic group, Sp$_{2n}[\mathbb{Z}]$. 
The $2n-$dimensional symplectic maps will allow couplings of various degrees of locality and strength. 

To begin with, we shall show that these maps can be constructed as iteration matrices of $n$ coupled Fibonacci sequences. 

We start from the relation between the ACM and the Fibonacci sequence of integers.

The Fibonacci sequence is one of the integer sequences, which has been studied, for a long time and there are journals dedicated to its properties and their applications. 

The definition is given by the relations
\begin{equation}
\label{fibinacci_seq}
\begin{array}{l}
f_0=0; f_1=1\\
f_{m+1}=f_m+f_{m-1}\\
\end{array}
\end{equation}
which can be written  in matrix form
\begin{equation}
\label{matrixfib}
\left(\begin{array}{c} f_m\\f_{m+1}\end{array}\right)=\underbrace{\left(\begin{array}{cc} 0 & 1 \\ 1 & 1\end{array}\right)}_{\sf A}\left(\begin{array}{c} f_{m-1} \\ f_m\end{array}\right)
\end{equation}
The matrix ${\sf A}$ is not a symplectic matrix, but it satisfies 
\begin{equation}
\label{nonsymplA}
{\sf A}^\mathrm{T}{\sf J}{\sf A}=-{\sf J}
\end{equation}
for $n=1$. 

We remark that the  Arnol'd cat map, acting on the torus $\mathbb{T}^2$, can be written as 
\be
\label{ArnoldCM}
\left(\begin{array}{cc} 1 & 1 \\ 1 & 2\end{array}\right) = {\sf A}^2
\ee
Eq.~(\ref{nonsymplA}) implies that ${[\sf A}^2]^\mathrm{T}{\sf J}{\sf A}^2={\sf J}$, therefore that ${\sf A}^2$ is symplectic. 
 It's possible to generalize the Fibonacci sequence in the  following way:
\begin{equation}
\label{kfibrec}
g_{m+1}=kg_m+g_{m-1}
\end{equation}
with $g_0=0$ and $g_1=1$ and $k$ is a positive integer.  This is known as the ``$k-$Fibonacci'' sequence~\cite{Horadam}. 

We may solve eq.~(\ref{kfibrec}) by $g_m\equiv C\rho^m.$ The characteristic equation for $\rho$  reads
\be
\label{kfibseq}
\rho^2-k\rho-1=0\Leftrightarrow\rho_\pm(k)=\frac{k\pm\sqrt{k^2+4}}{2}
\ee
and express $g_m$ as a linear combination of the $\rho_\pm$, upon taking into account the initial conditions:
\be 
\label{kfibsol}
g_m = A_+\rho_+(k)^m + A_-\rho_-(k)^m=\frac{\rho_+(k)^m-(-)^m\rho_+(k)^{-m}}{\sqrt{k^2+4}}
\ee
In matrix form
\be
\label{kfibinacci}
\left(\begin{array}{c} g_m\\ g_{m+1}\end{array}\right)=\underbrace{\left(\begin{array}{cc} 0 & 1\\ 1 & k\end{array}\right)}_{{\sf A}(k)}\left(\begin{array}{c} g_{m-1}\\g_m\end{array}\right)
\ee
Similarly as for the usual Fibonacci sequence, we may show, by induction, that 
\be
\label{Akn}
{\sf A}(k)^{m}=\left(\begin{array}{cc} g_{m-1} & g_m \\ g_m & g_{m+1}\end{array}\right)
\ee
We remark that $\mathrm{det}\,{\sf A}(k)^m=(-)^m$, for $m=1,2,\ldots$ and that $\lim_{m\to\infty}\,(g_{m+1}/g_m)=\rho_+(k)$, which, for $k=1$, is the golden ratio, for $k=2$ is the silver ratio and, for $k>2$ are generalizations thereof. 

The greatest   Lyapunov exponent of  the $k-$Arnol'd cat map, ${\sf A}(k)^2$, i.e.  $\lambda_+(k)\equiv \log\,\rho_+(k)^2$ is  an increasing function of $k$.

Coupled Fibonacci sequences have been considered in the literature, for instance in~\cite{rathore2012generalized}. However, in these papers the possible applications to Hamiltonian dynamics,  were not the topic of interest and moreover the corresponding maps were not symplectic.

After this review of the single Arnol'd cat map, we proceed to the study of how, many such maps, can interact, in a way such that the evolution matrix is an element of Sp$_{2n}[\mathbb{Z}].$ To this end, we shall, once more, use the correspondence between the Arnol'd cat maps and the Fibonacci sequences. 

We start by coupling $n=2$ Fibonacci sequences, $\{f_m\}$ and $\{g_m\}$  (but we write the expressions in a way that generalizes immediately to arbitrary $n$):
\begin{equation}
\label{Fibonacci}
\begin{array}{l}
\displaystyle
f_{m+1}=a_1f_m+b_1f_{m-1}+c_1g_m+d_1g_{m-1}\\
\displaystyle
g_{m+1}=a_2g_m+b_2g_{m-1}+c_2f_m+d_2f_{m-1}\\
\end{array}
\end{equation}
where the $a_i,b_i,c_i,d_i,i=1,2$ are integers,  $f_0=0=g_0$ and $f_1=1=g_1$ are the initial conditions and $m=1,2,3,\ldots$. 
 
In matrix form, these read 
\begin{equation}
\label{2Fibiter}
X_{m+1}\equiv
\left(\begin{array}{l} f_m\\g_m\\f_{m+1}\\g_{m+1}\end{array}\right)=
\left(\begin{array}{cccc}
0 & 0 & 1 &0\\
0 & 0 & 0 &1\\
b_1 & d_1 & a_1 & c_1 \\
d_2 & b_2 & c_2 & a_2
\end{array}
\right)
\underbrace{
\left(\begin{array}{l} f_{m-1}\\g_{m-1}\\f_{m}\\g_{m}\end{array}\right)}_{X_m}
\end{equation}
Let us define the 2$\times$2 matrices
\begin{equation}
\label{blocks}
\begin{array}{ccc}
{\sf D}\equiv\left(\begin{array}{cc} b_1 & d_1 \\ d_2 & b_2 \end{array}\right) & & 
{\sf C}\equiv\left(\begin{array}{cc} a_1 & c_1 \\ c_2 & a_2\end{array}\right)
\end{array}
\end{equation}
in terms of which the one--time--step evolution equation~(\ref{2Fibiter})  can be written in block form as 
\begin{equation}
\label{block1}
X_{m+1}=
\left(\begin{array}{cc} 0_{n\times n} & I_{n\times n} \\ {\sf D} & {\sf C}\end{array}\right) X_m
\end{equation}
In analogy with the case of a single Fibonacci sequence and its relation with the Arnol'd cat map, we impose the constraint (cf. ~(\ref{nonsymplA})) 
\begin{equation}
\label{symplectic1}
\left(\begin{array}{cc} 0_{n\times n} & I_{n\times n} \\ {\sf D} & {\sf C}\end{array}\right)^\mathrm{T}
{\sf J}
\left(\begin{array}{cc} 0_{n\times n} & I_{n\times n} \\ {\sf D} & {\sf C}\end{array}\right) = -{\sf J}
\end{equation}
This condition implies that 
\begin{equation}
\label{solsAB}
\begin{array}{ccc}
{\sf D} = I_{n\times n} & & 
{\sf C} = {\sf C}^\mathrm{T}
\end{array}
\end{equation}
Therefore $a_1=k_1$, $a_2=k_2$, $c_1 = c_2 = c$. This implies, in particular, that the coupling between the sequences is the same for both, in order for the square of the evolution matrix to be symplectic. As we shall explain below, the generalization of this property is that, for more sequences, the corresponding coupling matrix must be symmetric. This is  important in order to define  chains of interacting $k-$Arnol'd cat maps. 

 In terms of these parameters, the recursion relations take the form
\begin{equation}
\label{New2Fib}
\begin{array}{l}
\displaystyle
f_{m+1}=k_1f_m+f_{m-1}+cg_m\\
\displaystyle
g_{m+1}=k_2g_m+g_{m-1}+cf_m\\
\end{array}
\end{equation}
and can be identified as describing a particular  coupling between a $k_1-$ and a $k_2-$Fibonacci sequence. This particular coupling is determined by the condition that the square of the evolution matrix is an element of Sp$_4[\mathbb{Z}]$:
\begin{equation}
\label{Msquareevol}
{\sf A}=\left(\begin{array}{cc} 
0 & 1 \\ 1 & {\sf C}
\end{array}\right)\Rightarrow
{\sf M}={\sf A}^2=
\left(\begin{array}{cc} 
1 & {\sf C} \\ {\sf C} & 1+{\sf C}^2
\end{array}\right)=\left(\begin{array}{cc} I & 0 \\ {\sf C} & I\end{array}\right)\left(\begin{array}{cc} I & {\sf C} \\ 0 & I\end{array}\right)
\end{equation}

The generalization to a chain of $n$ $k-$Fibonacci sequences, with nearest--neighbor couplings,  proceeds as follows:  We choose two diagonal matrices, of positive integers, ${\sf K}_{IJ}=K_I\delta_{IJ}$ and ${\sf G}_{IJ}=G_I\delta_{IJ}$, with $I,J=1,2,\ldots,n$. 

We now define the translation operator, ${\sf P}$ along a closed chain, by ${\sf P}_{I,J}=\delta_{I-1,J}\,\mathrm{mod}\,n$. The periodicity is expressed by the fact that 
${\sf P}^n=I_{n\times n}$. 
Moreover, ${\sf P}$ is orthogonal, since ${\sf P}{\sf P}^\mathrm{T}=I_{n\times n}$.  

Now we can define the coupling matrix for  $n$  sequences as
\begin{equation}
\label{evolAn}
{\sf C} = {\sf K} + {\sf P}{\sf G} + {\sf G}{\sf P}^\mathrm{T}
\end{equation}

The corresponding $2n\times 2n$ evolution matrix, ${\sf A}$ is given by 
\begin{equation}
\label{Aevol2n}
{\sf A} = \left(\begin{array}{cc} 
0_{n\times n} & I_{n\times n} \\ I_{n\times n} & {\sf C}
\end{array}\right)
\end{equation}
and satisfies the relation ${\sf A}^\mathrm{T}{\sf J}{\sf A}=-{\sf J}$.
Its square, 
\begin{equation}
\label{Mevol2n}
{\sf M}={\sf A}^2=\left(\begin{array}{cc} 
I_{n\times n} & {\sf C} \\ {\sf C} & I_{n\times n}+{\sf C}^2
\end{array}\right)
\end{equation}
therefore satisfies the relation  ${\sf M}^\mathrm{T}{\sf J}{\sf M}={\sf J}$, showing that 
${\sf M}\in\mathrm{Sp}_{2n}[\mathbb{Z}]$. Since ${\sf A}$ is symmetric, (from the property that ${\sf C}={\sf C}^\mathrm{T}$), ${\sf M}$ is positive definite and its eigenvalues come in pairs, $(\rho,1/\rho)$, with $\rho > 1$ (and the corresponding eigenvectors are orthogonal).  This property implies that, for all matrices ${\sf K}$ and ${\sf G}$ this system of coupled maps is hyperbolic.

An important special case arises if we impose  translation invariance along the chain, i.e.  $K_I=K$ and $G_I = G$ for all $I=1,2,\ldots,n$.

Let us now consider the case of the open chain. The only change involves the operator ${\sf P}$, which, now, must be defined as ${\sf P}_{IJ}=\delta_{I-1,J}$, for $I,J=1,2,\ldots,n$. Due to the absence of the mod $n$ operation, the ``far non--diagonal'' (upper right and lower left) elements are, now, zero. This express the property that the $n-$th Fibonacci is not coupled to the first one (and vice versa). 

For both, closed or open, chains, we observe certain algebraic properties of the evolution matrix, ${\sf A}$. 

The $k-$Fibonacci sequence has the important property that the elements of the matrix  ${\sf A}(k)^m$ are arranged in columns of consecutive pairs of the sequence. We shall show that this property can be generalized for $n$ interacting $k-$Fibonacci sequences as follows:

\begin{theorem}
The $m-$th power of the evolution matrix, ${\sf A}$ (cf. eq.~(\ref{Aevol2n})) can be written as 
\begin{equation}
\label{Amevol2n}
{\sf A}^m=\left(\begin{array}{cc} {\sf C}_{m-1} & {\sf C}_m\\ {\sf C}_m & {\sf C}_{m+1}\end{array}\right) 
\end{equation}
where ${\sf C}_0=0_{n\times n}$, ${\sf C}_1=I_{n\times n}$ and ${\sf C}_{m+1}={\sf C}{\sf C}_m+{\sf C}_{m-1}$, with $m=1,2,3,\ldots$.   
This matrix recursion relation generalizes to matrices the $k-$Fibonacci sequence for numbers. It  holds for any  matrix, ${\sf C}$ and, in particular for the (symmetric, integer) matrix ${\sf C}$, defined by eq.~(\ref{evolAn}). The solution to this matrix recursion relation is given in terms of the Fibonacci polynomials, ${\sf F}_m(x),$ with argument $x={\sf C},$ i.e. ${\sf C}_m = {\sf F}_m({\sf C}).$~\cite{axenides2023arnol}.
\end{theorem}
\begin{proof}
The proof is by induction. For $m=1$ it is true, by definition. If we assume it holds for $m>1$, then, by the relation ${\sf A}^{m+1}={\sf A}\cdot{\sf A}^m$, we immediately establish that it holds for $m+1$. 
\end{proof}
Having constructed a large class of symplectic many--body maps, that describe the dynamics of $n$ $k-$Arnol'd cat maps, we wish to understand their chaotic behavior, as they act on any initial condition of $\mathbb{T}^{2n},$ for $n>1.$ 

These particular symplectic matrices generalize the chaotic behavior of one Arnol'd cat map, acting on $\mathbb{T}^2,$ to that of $n$ maps, acting on $\mathbb{T}^{2n}.$ 
Indeed, starting with initial conditions $\bm{x}_0\in\mathbb{T}^{2n}\neq\bm{0},$ which have irrational components, the evolution matrix ${\sf M},$ as it acts on $\bm{x}_0,$ 
defines an orbit,$\bm{x}_m=\bm{x}_0{\sf M}^m,$ that will, in the limit $m\to\infty,$ cover the whole torus $\mathbb{T}^{2n};$ the full phase space is the attractor of the map. 
This means that the map is ergodic. 

Furthermore, this map has the, additional, property of being {\em strongly chaotic}, which means that it has positive Kolmogorov-Sinai entropy. It is possible to tune the parameters $K$ and $G$ (in the translation invariant case) so that no Lyapunov exponent is equal to  zero (which would imply the existence of a conservation law). In this cse the map is {\em maximally} hyperbolic, i.e. it is an ``Anosov C-system''. 

Another quantitative measure of the ``long time'' chaotic behavior of the orbits  of a hyperbolic map is provided by the properties of the time correlation functions of ``observables'', i.e. functions on the phase space, $\mathbb{T}^{2n}.$  These properties define the {\em mixing behavior} of the map and can be expressed by the fact that the connected correlation functions decay to zero, for long times. This, thus, raises the problem of computing them and, in particular whether they decay exponentially, thereby defining the spectrum of mixing times. 

On the other hand, another way to look at the chaotic properties of a dynamical system, is to consider, if possible, the full set of its unstable periodic orbits. For large periods, the corresponding  periodic orbit should approach (this is called ``shadowing'') the chaotic orbits, which fill the phase space~\cite{auerbach1987exploring}. Fortunately, for our system, of $n$ interacting Arnol'd cat maps, the full set of unstable periodic orbits is produced by all the {\em rational} points of $\mathbb{T}^{2n},$ taken as initial conditions. So the problem reduces to finding periodic orbits, with ``very large'' periods. This problem is difficult, because the periods are random functions of the--common--denominator, $N,$ of the rational points, that are initial conditions~\cite{axenides2023arnol}.    
 
In the following section we shall, thus, introduce the notion of mixing and  the mixing time for ergodic systems and  we shall present the method for obtaining a bound on the mixing time for the case of the single cat map; then we shall do the same calculation for two, symplectically coupled, cat maps, to see what are the issues that arise. Finally, we shall present the generalization to $n$ symplectically coupled maps, where, in addition to the dependence on the coupling, the issue of the range becomes of interest.

\section{The mixing time for one cat map}\label{ACM1mix}

Mixing is the ubiquitous phenomenon of  blending together distinct many body matter systems  from an initial inhomogeneous state to a final homogeneous uniform configuration ~\cite{cornfeld2012ergodic}. The mixing time is the characteristic transient period it takes for an initial local perturbation to delocalize and spread in a many-body system  attaining a uniform and homogenized  final state. 
What sets the scale of mixing time as well as its precise determination is both a conceptual and a computational challenge.

The development of mathematical methods associated with the Ergodic Theory have been brought into prominence  ~\cite{sinaui1976introduction,cornfeld2012ergodic} in the study of "Ergodic Mixing" through the study of stochasticity in measure-preserving dynamical systems. 

For a discrete-time  dynamical system, $ T : M\ra M $ ,  which preserves a probability measure $ \mu$  a "strong  mixing"' condition can be formulated  for two sets of points A and B on a constant energy surface $ E$ , as follows: 
\begin{equation}
\mu\left( A \cap T^{n}\left(B\right)\right)
\xrightarrow{n\longrightarrow\infty}
{\mu}\left( A \right) \mu\left( B\right) 
\end{equation}
 Intuitively the mixing condition states that a dynamical system is strongly mixing whenever any two observables $A$ and $B$  which occur at separate time instances, specified by the action of $T^{n}$ on $B,$
 become independent in the infinite time separation limit $ n \ra \infty $.

Equivalently its standard diagnostic quantifier  is the decay of the correlation function for any pair of observable functions $f , g  : M \rightarrow \mathbb{C}$  
\begin{equation}
C_{n} \left( f , g^\ast \right) \ = \ \int \ f \circ	T^{n} \cdot g^\ast d\mu \ - \ 
\left( \int f d\mu \right) \cdot \left( \int g^\ast d\mu \right)\hskip0.25truecm
\xrightarrow{n\longrightarrow\infty}0
\end{equation}
The correlation function $ C_{n} \left( f , g^\ast \right) $  or the self-correlation for a single function $ C_{n} \left( f,f^\ast\right) $ fall off to zero either polynomially, exponentially (rapid mixing) or even super exponentially in the long time limit $ n\ra \infty $~\cite{baladidecay2001,Pollicott2019}.

We shall show explicitly that the Arnol'd cat map on $\mathbb{T}^2$ is exponentially mixing and  that its mixing time is given by the expression
\begin{equation}
\label{mixingtimeACM}
\tau_\mathrm{mix}\mathrm{(Arnol'd)}=\frac{1}{\log\rho_+}
\end{equation}
where $\rho_+=(3+\sqrt{5})/2.$ This is, in fact equal to the inverse of the Kolmogorov--Sinai entropy of the system. This result is known and in the next section we shall generalize it to the case of symplectically coupled Arnol'd cat maps.

We start with some preparatory material: Any smooth and square integrable observable on $\mathbb{T}^2$ has a uniformly convergent Fourier series
\begin{equation}
\label{Fourierseriesf}
f(\bm{x})=\sum_{\bm{k}\in\mathbb{Z}\times\mathbb{Z}}\,c_{\bm{k}} e^{2\pi\mathrm{i}\bm{k}\cdot\bm{x}}
\end{equation}
The set of all these observables defines a Hilbert space, $H(\mathbb{T}^2),$ with inner product 
\begin{equation}
\label{innerpod}
\langle f,g\rangle=\int_{\mathbb{T}^2}\,d^2\bm{x}\,f(\bm{x})g^\ast(\bm{x})=\sum_{\bm{k}\in\mathbb{Z}\times\mathbb{Z}} c_{\bm{k}}\cdot d_{\bm{k}}^\ast
\end{equation}
where
\begin{equation}
\label{Fourierseriesg}
g(\bm{x})=\sum_{\bm{k}\in\mathbb{Z}\times\mathbb{Z}}\,d_{\bm{k}} e^{2\pi\mathrm{i}\bm{k}\cdot\bm{x}}
\end{equation}
and the norm is defined by 
\begin{equation}
\label{norm}
||f||^2=\int_{\mathbb{T}^2}\,d^2\bm{x}\,f(\bm{x})f^\ast(\bm{x})=\sum_{\bm{k}\in\mathbb{Z}\times\mathbb{Z}}\,|c_{\bm{k}}|^2
\end{equation}
For any $f,g\in H(\mathbb{T}^2)$ we have the Cauchy--Schwarz inequality
\begin{equation}
\label{CSineq}
\left|\langle f,g\rangle\right| \leq ||f|| ||g||
\end{equation}
With these standard preliminaries, we proceed with the evaluation of the correlation functions for the Arnol'd cat map. 
\begin{equation}
\label{corrfunsACM}
C_n(f,g^\ast)=\int\,d^2\bm{x}\,f(T^n\bm{x})g^\ast(\bm{x})-\int\,d^2\bm{x}\,f(\bm{x})\int\,d^2\bm{x}\,g^\ast(\bm{x})=\langle f\circ T^n,g^\ast\rangle-\langle f\rangle\langle g^\ast\rangle
\end{equation}

We follow the procedure sketched in ref.~\cite{de1995chaos}.  The idea is to choose as functions the eigenfunctions of a particular operator and show that, for these, $|C_n(f,g^\ast)|\leq\mathrm{const}\times e^{-n/\tau}.$ 

Let us recall the argument: The ACM has eigenvalues 
\begin{equation}
\label{eigenvalACM}
\rho_\pm=\frac{3\pm\sqrt{5}}{2}
\end{equation}
 and corresponding eigenvectors, 
 \begin{equation}
 \label{eigenfunACM}
 \bm{u}_\pm=\frac{1}{\sqrt{1+\rho_\pm}}\left(\begin{array}{c} 1 \\ \rho_\pm-1\end{array}\right)
 \end{equation}
 
 Therefore, $${\sf A}\bm{u}_\pm=\rho_\pm\bm{u}_\pm$$ 
 
 Since $\rho_+\rho_-=1$ and they're real (since the matrix is symmetric) $\rho_+>1$ and $\rho_-=1/\rho_+ < 1.$ 

Now, let us define the operators 
\begin{equation}
\label{Dpm}
D_\pm\equiv -\frac{\mathrm{i}}{2\pi}u_{\pm,I}\delta_{IJ}\partial_{x_J}
\end{equation}
Since the torus, $\mathbb{T}^2,$ is a compact manifold, the spectrum of the operators $D_\pm$ is discrete and can be labeled by two integers, $n_I.$ Indeed, it is straightforward to check that the functions 
\begin{equation}
\label{Dpmeigenfun}
e_{\bm{k}}(\bm{x})=e^{2\pi\mathrm{i}\bm{k}\cdot\bm{x}}
\end{equation}  
with $\bm{k}\in\mathbb{Z}\times\mathbb{Z}$ and $\bm{x}\in\mathbb{T}^2$ are eigenfunctions of $D_\pm:$
\begin{equation}
\label{Dpmeigenfun1}
D_\pm e_{\bm{k}}(\bm{x})=\bm{k}\cdot\bm{u}_\pm e_{\bm{k}}(\bm{x})
\end{equation}
Since the components of the $\bm{u}_\pm$ are irrational numbers, $\bm{k}\cdot\bm{u}_\pm\neq 0$ for any $\bm{k}\in\mathbb{Z}\times\mathbb{Z},$ $\bm{k}\neq(0,0),$  the inverses, $[D_\pm]^{-1},$ 
\begin{equation}
\label{Dpminv}
[D_\pm]^{-1} e_{\bm{k}}(\bm{x})=\frac{1}{\bm{k}\cdot\bm{u}_\pm}e_{\bm{k}}(\bm{x})
\end{equation} 
are well-defined.

We proceed below with the details of the evaluation of the correlation function $C_n(f,g^\ast)$ in a more explicit form. 

First we split the Fourier sums of the functions $f$ and $g$ as follows:
\begin{equation}
\label{CorrfunACMmodes}
\begin{array}{l}
\displaystyle
C_n(f,g^\ast)=\sum_{\bm{k},\bm{l}\in\mathbb{Z}\times\mathbb{Z}} c_{\bm{k}}d_{\bm{l}}^\ast\int_{\mathbb{T}^2}\,d^2\bm{x}\,e^{2\pi\mathrm{i}\bm{k}\cdot T^n\bm{x}} e^{-2\pi\mathrm{i}\bm{l}\cdot\bm{x}}-c_0d_0^\ast=\\
\displaystyle
\sum_{\bm{k}\neq(0,0)} c_{\bm{k}}d_0^\ast\int_{\mathbb{T^2}}\,d^2\bm{x}\,e^{2\pi\mathrm{i}\bm{k}\cdot T^n\bm{x}} +
\sum_{\bm{l}\neq(0,0)} c_0 d_{\bm{l}}^\ast\int_{\mathbb{T}^2}\,d^2\bm{x}\,e^{-2\pi\mathrm{i}\bm{l}\cdot\bm{x}} + 
\sum_{\bm{k}\neq(0,0),\bm{l}\neq(0,0)} c_{\bm{k}}d_{\bm{l}}^\ast\int_{\mathbb{T}^2}\,d^2\bm{x}\,e^{2\pi\mathrm{i}\bm{k}\cdot T^n\bm{x}}e^{-2\pi\mathrm{i}\bm{l}\cdot\bm{x}}
\end{array}
\end{equation}
The first two terms, that involve sums over $\bm{k}\neq(0,0)$ or $\bm{l}\neq (0,0)$ are zero (since the integrals are $\delta-$functions on $\bm{k}=(0,0)$ and $\bm{l}=(0,0)$ respectively), therefore only the last term survives:
\begin{equation}
\label{CorrfunACMmodes1}
C_n(f,g^\ast)=\sum_{\bm{k}\neq(0,0),\bm{l}\neq(0,0)} c_{\bm{k}}d_{\bm{l}}^\ast\int_{\mathbb{T}^2}\,d^2\bm{x}\,e^{2\pi\mathrm{i}\bm{k}\cdot T^n\bm{x}}e^{-2\pi\mathrm{i}\bm{l}\cdot\bm{x}}
\end{equation}
Therefore, the operators $D_\pm\equiv -(\mathrm{i}/(2\pi))\bm{u}_\pm\cdot\nabla_{\bm{x}}$ act on the function $f(\bm{x})=e^{2\pi\mathrm{i}\bm{k}\cdot({\sf A}^n\bm{x})}$
as
$$
D_\pm e^{2\pi\mathrm{i}\bm{k}\cdot({\sf A}^n\bm{x})}=\bm{u}_\pm\cdot\bm{k}\rho_\pm^n e^{2\pi\mathrm{i}\bm{k}\cdot({\sf A}^n\bm{x})}
$$
as can be checked by direct calculation. 

These calculations now lead to the following statement:
\begin{prop}
$$
\int_{\mathbb{T}^2}\,d^2\bm{x}\,D_\pm f({\sf A}^n\bm{x})\,D_\pm^{-1}\,g^\ast(\bm{x})=-\int_{\mathbb{T}^2}\,d^2\bm{x}\,f({\sf A}^n\bm{x})\,g^\ast(\bm{x})
$$
\label{prop1}
\end{prop}
which can be proved using  integration by parts.

Since any function on the torus can be expanded in plane waves, we deduce that the correlation function of any two functions on the torus, that are sufficiently smooth for their Fourier expansions (and those of their derivatives) to converge, will show mixing behavior with the same mixing time. 

Using Proposition~\ref{prop1}, this can be established as follows: 
\begin{equation}
\label{Cnfgastums}
\begin{array}{l}
\displaystyle
C_n(f,g^\ast)=\lambda_\pm^n\int_{\mathbb{T}^2}\,d^2\bm{x}\,\left(
\sum_{\bm{k}\neq(0,0)}\,c_{\bm{k}}\bm{k}\cdot\bm{u}_\pm e^{2\pi\mathrm{i}\bm{k}\cdot T^n\bm{x}}\right)
\left(
\sum_{\bm{l}\neq(0,0)}\,d_{\bm{k}}^\ast\frac{1}{\bm{l}\cdot\bm{u}_\pm}e^{-2\pi\mathrm{i}\bm{l}\cdot\bm{x}}
\right)
\end{array}
\end{equation}
We now apply the Cauchy-Schwarz inequality:
\begin{equation}
\label{CSineq}
|\langle f,g^\ast\rangle|^2\leq ||f||^2||g^\ast||^2
\end{equation}
to deduce that 
\begin{equation}
\label{CSinqmix}
|C_n(f,g^\ast)|^2\leq \rho_\pm^{2n}\int_{\mathbb{T}^2}\,d^2\bm{x}\,\left|\sum_{\bm{k}\neq(0,0)}\,c_{\bm{k}}\bm{k}\cdot\bm{u}_\pm e^{2\pi\mathrm{i}\bm{k}\cdot\bm{x}}\right|^2
\int_{\mathbb{T}^2}\,d^2\bm{x}\,\left|\sum_{\bm{l}\neq(0,0)}\,d_{\bm{k}}^\ast\frac{\bm{l}\cdot\bm{u}_\pm} e^{-2\pi\mathrm{i}\bm{l}\cdot\bm{x}}\right|^2
\end{equation}
The integrals do not depend on $n,$ so the best bound on how $|C_n(f,g^\ast)|$ vanishes, as $n\to\infty,$ is obtained in the form
\begin{equation}
\label{tmixACM}
|C_n(f,g^\ast)|\sim e^{-n\log\rho_+}=e^{-n\lambda_+}
\end{equation}
as $n\to\infty.$

\section{Multiple mixing for the single Arnol'd cat map}\label{multiplemix}

We now proceed to extend the notion of mixing to any number of observables in both the ergodic theory of measure preserving transformations as well as to their corresponding diagnostic quantifiers of decay of higher order correlation functions of observables. 
This was pioneered by Rokhlin~\cite{rokhlin1949endomorphisms} and has since been the focus of considerable activity (cf.~\cite{ryzhikov2024multiplemixing75years} for a recent review).

Let $T$ be the measure preserving transformation of a dynamical system $\Gamma$ on a phase space $\Sigma.$ 
The generalized 3-strong mixing condition for any three observables or measurable sets A, B , and C   take the form 
\begin{equation}
\label{3-mixing}
\lim_{m,n\ra \infty} \mu\left( A \cap T^{m} B \cap T^{ m + n } C \right) \, = \, \mu \left( A\right) \mu \left(B\right) \mu \left( C\right) 
\end{equation}

The strong mixing condition for  the more general case of $l-$fold mixing takes the form
\begin{equation}
\label{k-mixing}
\lim_{\alpha_{1}, \ldots ,\alpha_{k}\ra\infty} \mu \left( A_{1}\cap T^{\alpha_{1}} A_{2} \ldots T^{\sum_{i=1}^{k-1} \alpha_{i}}
 A_{i} \right) \, = \, \prod_{i=1}^{k} \mu\left ( A_{i} \right) 
\end{equation} 

In the equivalent language of correlations functions we define the $l-$fold mixing correlation function of $l+1$ observables, $f_i(\bm{x}), i=1,2,\ldots,l+1,$ $l=1,2,\ldots,$ $\bm{x}\in\Sigma$ as follows: 
\begin{equation}
\label{Cmultiple}
\begin{array}{l}
\displaystyle
C_{n_1,n_2,\ldots,n_l}(f_1,\ldots,f_{l+1})=\int_\Sigma\,d\mu(\bm{x})f_1(\bm{x})f_2(T^{n_1}\bm{x})f_3(T^{n_1+n_2}\bm{x})\cdots f_{l+1}(T^{n_1+n_2+\cdots+n_l}\bm{x})-\\
\displaystyle
\hskip5truecm
\prod_{i=1}^{l+1}\,\int_\Sigma\,d\mu(\bm{x})\,f_i(\bm{x})
\end{array}
\end{equation}
where $n_1,n_2,\ldots,n_l=1,2,\ldots$

We say that the dynamical system $(\Gamma, T,\Sigma)$ exhibits $l-$fold  mixing iff
\begin{equation}
\label{lfoldmixing}
C_{n_1,n_2,\ldots,n_l}(f_1,\ldots,f_{l+1})\to 0
\end{equation}
as $(n_1,n_2,\ldots,n_l)\to\infty,$ for observables $\left\{f_i(\bm{x})\right\}_{i=1,\ldots,l+1}$ which are smooth enough and square integrable, as well as all their derivatives.

We are interested in the case when $T$ is the Arnol'd cat map and $\Sigma=\mathbb{T}^2,$ the two-dimensional torus with radii equal to 1 and measure $d\mu(\bm{x})=d^2\bm{x}.$ In this section we shall show that the ACM exhibits $l-$fold mixing for every integer $l=1,2,\ldots$ and compute the corresponding mixing times, generalizing the calculations of the previous section. 

To this end, we expand the $f_i(\bm{x})$ in Fourier series,
\begin{equation}
\label{fiFourier}
f_i(\bm{x})=\sum_{\bm{k}\in\mathbb{Z}\times\mathbb{Z}} c_{\bm{k}}^{(i)} e^{2\pi\mathrm{i}\bm{k}\cdot\bm{x}}
\end{equation}
and extract the constant part:
\begin{equation}
\label{constfipart}
f_i(\bm{x})=c_{\bm{0}}^{(i)}+\widetilde{f}_i(\bm{x})
\end{equation}
where 
\begin{equation}
\label{c0i}
c_{\bm{0}}^{(i)}=\int_{\mathbb{T}^2}\,d^2\bm{x}\,f_i(\bm{x})
\end{equation}
and
\begin{equation}
\label{widetildefi}
\widetilde{f}_i(\bm{x})=\sum_{\bm{k}\neq(0,0)} c_{\bm{k}}^{(i)} e^{2\pi\mathrm{i}\bm{k}\cdot\bm{x}}
\end{equation} 
We apply this decomposition only to $f_1(\bm{x})$ and obtain the following expression for 
$C_{n_1,n_2,\ldots,n_l}(f_1,\ldots,f_{l+1}):$
\begin{equation}
\label{Cmultiple1}
\begin{array}{l}
\displaystyle
C_{n_1,n_2,\ldots,n_l}(f_1,\ldots,f_{l+1})=c_{\bm{0}}^{(1)}\int_\Sigma\,d^2\bm{x}\,f_2(T^{n_1}\bm{x})\cdots f_{l+1}(T^{n_1+\cdots+n_l}\bm{x}) +\\
\displaystyle
\hskip2truecm
\int\,d^2\bm{x}\,\widetilde{f}_1(\bm{x})f_2(T^{n_1}\bm{x})\cdots f_{l+1}(T^{n_1+\cdots+n_l}\bm{x}) - \prod_{i=1}^{l+1}\,c_{\bm{0}}^{(i)}
\end{array}
\end{equation}
The first integral in this expression can be rewritten, upon performing the change of variables, $\bm{x}\to\bm{y}=T^{n_1}\bm{x},$ that leaves the measure invariant, as
\begin{equation}
\label{Cmultiple2}
\begin{array}{l}
\displaystyle
\int_\Sigma\,d^2\bm{x}\,f_2(T^{n_1}\bm{x})\cdots f_{l+1}(T^{n_1+\cdots+n_l}\bm{x})=\int_\Sigma\,d^2\bm{x}\,f_2(\bm{x})\cdots f_{l+1}(T^{n_2+\cdots+n_l}\bm{x})=\\
\displaystyle
\hskip7truecm
C_{n_2,n_3,\ldots,n_l}(f_2,f_3,\ldots,f_{l+1}) + \prod_{i=2}^{l+1}c_{\bm{0}}^{(i)}
\end{array}
\end{equation}
whence we deduce the ``recurrence relation''
\begin{equation}
\label{Cmultiple3}
C_{n_1,\ldots,n_l}(f_1,f_2,\ldots,f_{l+1})=c_{\bm{0}}^{(1)}C_{n_2,n_3,\ldots,n_l}(f_2,f_3,\ldots,f_{l+1})+\int\,d^2\bm{x}\,\widetilde{f}_1(\bm{x})f_2(T^{n_1}\bm{x})
\end{equation}
What is noteworthy is that the product of the constant terms has been eliminated. 

If we now assume that $C_{n_2,\ldots,n_l}(f_2,\ldots, f_{l+1})$ tends to 0 as $n_2,\ldots,n_l\to\infty,$ it remains to prove that the last term in eq.~(\ref{Cmultiple3}), also, tends to 0, as $n_1,n_2,\ldots,n_l\to\infty.$

To this end we shall employ the method of the previous section, making use of the self-adjoint differential operator,
\begin{equation}
\label{Dops}
D_{\bm{u}}=\frac{1}{2\pi\mathrm{i}}\bm{u}\cdot\bm{\partial}
\end{equation} 
where $\bm{\partial}=(\partial_{x_1},\partial_{x_2})$  and $\bm{u}=(u_1,u_2)$ is the eigenvector of the ACM that corresponds to the largest eigenvalue of the ACM. The inverse differential operator, $[D_{\bm{u}}]^{-1}$ acts on the complex functions of the torus that don't have a constant term. This is the reason it's useful to extract it by writing the function $f_1(\bm{x})$ as $f_1(\bm{x})=c_{\bm{0}}^{(1)}+\widetilde{f}_1(\bm{x}).$ 

We notice, now,  that
\begin{equation}
\label{property}
I_{n_1,\ldots,n_l}\equiv\int_{\mathbb{T}^2}\,d^2\bm{x}\,\widetilde{f}_1(\bm{x})f_2(T^{n_1}\bm{x})\cdots f_{l+1}(T^{n_1+\ldots+n_l}\bm{x})=
-\int_{\mathbb{T}^2}\,d^2\bm{x}\,\left([D_{\bm{u}}]^{-1}\widetilde{f}_1(\bm{x})\right)D_{\bm{u}}\left( f_2\cdots f_{l+1}\right)
\end{equation} 
where 
\begin{equation}
\label{Duf}
D_{\bm{u}}f_i(\bm{x})=\sum_{\bm{k}\in\mathbb{Z}\times\mathbb{Z}}\,c_{\bm{k}}^{(i)}(\bm{k}\cdot\bm{u}) e^{2\pi\mathrm{i}\bm{k}\cdot\bm{x}}
\end{equation}
for $i=2,\ldots,l+1$ and
\begin{equation}
\label{Duinvf}
D_{\bm{u}}^{-1}\widetilde{f}_1(\bm{x})=\sum_{\bm{k}\in\mathbb{Z}\times\mathbb{Z}}\,c_{\bm{k}}^{(1)}\frac{1}{\bm{k}\cdot\bm{u}} e^{2\pi\mathrm{i}\bm{k}\cdot T^n\bm{x}}
\end{equation}
as well as 
\begin{equation}
\label{DufoTn}
D_{\bm{u}}f_i(T^n\bm{x})=\lambda^n\sum_{\bm{k}\in\mathbb{Z}\times\mathbb{Z}}\,c_{\bm{k}}^{(i)}(\bm{k}\cdot\bm{u}) e^{2\pi\mathrm{i}\bm{k}\cdot\bm{x}}
\end{equation}
We can now evaluate the action of $D_{\bm{u}}$ on the product $f_2\cdots f_{l+1};$ we find 
\begin{equation}
\label{Duprodf}
\begin{array}{l}
\displaystyle
D_{\bm{u}}\left[
f_2(T^{n_1+n_2}\bm{x})f_3(T^{n_1+n_2+n_3}\bm{x})\cdots f_{l+1}(T^{n_1+n_2+\cdots+n_{l+1}}\bm{x})
\right]=\\
\displaystyle
\sum_{m=2}^{l+1}\,F_{l,m}(\bm{x})D_{\bm{u}}f_m(T^{n_1+n_2+\cdots+n_m}\bm{x})
\end{array}
\end{equation}
where 
\begin{equation}
\label{Flm}
F_{l,m}(\bm{x})=\prod_{i=2,i\neq m}^{l+1}\,f_i(T^{n_1+n_2+\cdots+n_i}\bm{x})
\end{equation}
These imply that eq.~(\ref{Duprodf}) can be written as
\begin{equation}
\label{Duprodf1}
\begin{array}{l}
\displaystyle
D_{\bm{u}}\left[
f_2(T^{n_1+n_2}\bm{x})f_3(T^{n_1+n_2+n_3}\bm{x})\cdots f_{l+1}(T^{n_1+n_2+\cdots+n_{l+1}}\bm{x})
\right]=\\
\displaystyle
\sum_{m=2}^{l+1}\,F_{lm}(\bm{x}) \rho^{n_1+n_2+\cdots+n_m}f_m'(T^{n_1+n_2+\cdots+n_m}\bm{x})
\end{array}
\end{equation}
In order to obtain the desired result, that the integral in eq.~(\ref{property}) does vanish--and, what's much more interesting, how does it vanish--in the limit $n_2,n_3,\ldots,n_l\to\infty-$we apply the triangle and the Cauchy--Schwarz inequalities repeatedly:
\begin{equation}
\label{CStrinproperty}
\begin{array}{l}
\displaystyle
| I_{n_1,n_2,\ldots,n_l}|^2\leq || D_{\bm{u}}^{-1}\widetilde{f}_1||^2 || D_{\bm{u}}[f_2\cdots f_{l+1}] ||^2\leq || D_{\bm{u}}^{-1}\widetilde{f}_1||^2
\sum_{m=2}^{l+1}\,\rho^{2(n_1+\cdots+n_m)}|| F_{lm}(\bm{x})||^2 || f_m'||^2
\end{array}
\end{equation}
The norm of $F_{lm}$ can be bounded as
\begin{equation}
\label{CStrinproperty}
\begin{array}{l}
\displaystyle
|| F_{lm}(\bm{x})||^2\leq \prod_{i=2,i\neq m}^{l+1}\,|| f\circ T^{n_1+n_2+\cdots+n_{i-1}}\bm{x}||^2
\end{array}
\end{equation}
We notice that due to the measure preserving properties of $T$ the various factors do not depend on $n_1,n_2,\ldots,n_l:$
\begin{equation}
\label{masurepreservT}
|| f\circ T^{n_1+n_2+\cdots+n_{i-1}}\bm{x}||^2=\int\,d^2\bm{x}\,| f_i(T^{n_1+n_2+\cdots+n_{i-1}}\bm{x})|^2=\int\,d^2\bm{x}\,|f_i(\bm{x})|^2
\end{equation}
The same occurs for the $||f_m'||^2.$

Collecting together all terms we obtain the effective bound using the  eigenvector $\bm{u}_-,$ corresponding to the eigenvalue $\rho_-< 1,$ therefore
\begin{equation}
\label{bound}
| I_{n_1,\ldots, n_l}|^2 \leq \sum_{m=1}^{l+1} d_m\,\rho_-^{2(n_1+n_2+\cdots+n_m)}
\end{equation}
where the coefficient $d_m$ contains all numerical factors. 

To leading order, therefore, we find that 
\begin{equation}
\label{bound1}
| I_{n_1,\ldots, n_l}|^2\to d_1\rho_-^{2n_1}=d_1 e^{-2n_1\log\lambda_+}
\end{equation}
thereby completing the induction hypothesis.

That this term, indeed, vanishes, in the long time limit, guarantees the ``sufficient'' part of Rokhlin's conjecture, that 2-mixing induces $l-$fold mixing for any $l>2.$ 
  
\section{The mixing time for $n$ symplectically coupled  Arnol'd cat maps}\label{CACM}
It is, now, interesting to examine the mixing properties of the system of $n>1$ coupled Arnol'd cat maps. It is known that hyperbolic (Anosov) linear maps on compact phase spaces exhibit strong, as well as $l-$fold (for all $l=1,2,\ldots$), mixing. In this section, therefore, we shall focus on the effects the coupling  has on the mixing properties of such maps, generalizing the calculation for the single map, that was studied in section~\ref{ACM1mix} (the case of $l-$fold mixing is a straightforward generalization). 

The idea is to use the differential operators,
\begin{equation}
\label{diffopsDu}
D_{\bm{u}}^i=\frac{1}{2\pi\mathrm{i}}\bm{u}^i\cdot\bm{\partial}=\frac{1}{2\pi\mathrm{i}}u_a^i\partial_a
\end{equation}
that are the generalization for $n$ degrees of freedom of the differential operators~(\ref{Dops}) that were used for computing the mixing time of one cat map. Here 
$\partial_a=\partial/\partial\bm{x}_a$ and $a=1,2,\ldots,2n$ labels the point on the $2n-$dimensional torus. $\bm{u}^i$ is an eigenvector of the evolution operator
\begin{equation}
\label{MevolnACM}
{\sf M}=\left(\begin{array}{cc} I & {\sf C}\\ {\sf C} & I + {\sf C}^2\end{array}\right)
\end{equation}
where the symmetric, $n\times n$ integer matrix ${\sf C}$ is parametrized as
\begin{equation}
\label{Cmatrix}
{\sf C}=K+G({\sf P}+{\sf P}^\mathrm{T})
\end{equation}
for the case of the closed chain of $n$ maps and $K$ and $G$ are diagonal matrices. 

As for the single map, we are interested in bounding the correlation function
\begin{equation}
\label{corrfunnmaps}
C_r(f,g^\ast)=\int_{\mathbb{T}^{2n}}\,d^{2n}\bm{x}\,f({\sf M}^r\bm{x})g^\ast(\bm{x})-\int_{\mathbb{T}^{2n}}\,d^{2n}\bm{x}\,f(\bm{x})\int\,d^{2n}\,g^\ast(\bm{x})
\end{equation}
for $f,g\in L^2(\mathbb{T}^{2n}),$ which have continuous {\em and} square integrable partial derivatives.

 We proceed as follows: We observe that to get the strong mixing property as well as the rate of convergence (the mixing time), it is enough to use a straightforward generalization of the method used in for the single map, {\em viz.} we write 
 \begin{equation}
 \label{constantfgsep}
 \begin{array}{l}
 \displaystyle
 f(\bm{x})=\sum_{\bm{k}\in\mathbb{Z}^{2n}}\,c_{\bm{k}}^f\,e^{2\pi\mathrm{i}\bm{k}\cdot\bm{x}}\equiv c_{\bm{0}}^f+\widetilde{f}(\bm{x})\\
 \displaystyle
 g(\bm{x})=\sum_{\bm{k}\in\mathbb{Z}^{2n}}\,c_{\bm{k}}^g\,e^{2\pi\mathrm{i}\bm{k}\cdot\bm{x}}\equiv c_{\bm{0}}^g+\widetilde{g}(\bm{x})\\
 \end{array}
\end{equation}  
whence we, immediately, find that 
\begin{equation}
\label{Crfg}
\begin{array}{l}
\displaystyle 
C_r(f,g^\ast)= {c_0^{(g)}}^\ast\int\,d^{2n}\bm{x} f({\sf M}^r\bm{x}) + \int\,d^{2n}\bm{x}{\widetilde{g}(\bm{x})}^\ast f({\sf M}^r\bm{x})-c_0^{(f)}{c_0^{(g)}}^\ast=\\
\displaystyle
{c_0^{(g)}}^\ast\int\,d^{2n}\bm{x}\,\widetilde{f}({\sf M}^r\bm{x}) + \int\,d^{2n}\,\bm{x}\,{\widetilde{g}}^\ast(\bm{x})f({\sf M}^r\bm{x}) = \int\,d^{2n}\bm{x}\,{\widetilde{g}}^\ast(\bm{x})f({\sf M}^r\bm{x})
\end{array}
\end{equation}
Now we  carry out the same procedure as for the case of one map, by acting now with the   product of all the $D_u^{(i)},$ $i=1,\ldots,n.$ Upon integrating by parts as for the case of one map, we find, this time, that
\begin{equation}
\label{Crfg1}
\begin{array}{l}
\displaystyle
C_r(f,g^\ast)=(-)^n\int\,d^{2n}\bm{x}\,[D_{u_-}^{(1)}]^{-1}\cdots [D_{u_-}^{(n)}]^{-1}{\widetilde{g}(\bm{x})}^\ast
D_{u_-}^{(1)}\cdots D_{u_-}^{(n)}f({\sf M}^r\bm{x})=\\
\displaystyle
(-)^n\int\,d^{2n}\bm{x}\,\sum_{\bm{k}\neq\bm{0}}\,\frac{1}{\displaystyle\prod_{i=1}^n\,\bm{u}_-^{(i)}\cdot\bm{k} }{c_{\bm{k}}^{g}}^\ast e^{-2\pi\mathrm{i}\bm{k}\cdot\bm{x}}
\displaystyle\prod_{i=1}^n\,\lambda_-^{(i)}\sum_{\bm{l}\in\mathbb{Z}^{2n}}\,c_{\bm{l}}^{f} \bm{u}_-^{(i)}\cdot\bm{l}\,e^{2\pi\mathrm{i}\bm{l}\cdot  {\sf M}^r\bm{x}}
\end{array}
\end{equation}
Using the Cauchy-Schwarz inequality, we obtain the bound 
\begin{equation}
\label{Crfgbound}
\left|C_r(f,g^\ast)\right|\leq \left(\prod_{i=1}^n\,\rho_-^{(i)}\right)^r||\widetilde{\widetilde{g}}|| ||\widetilde{\widetilde{f}}||
\end{equation} 
where
\begin{equation}
\label{doubletildes}
\begin{array}{l}
\displaystyle
\widetilde{\widetilde{g}}=\sum_{\bm{k}\neq\bm{0}}\,\frac{1}{\displaystyle \prod_{i=1}^n\,\bm{u}_-^{(i)}\cdot\bm{k}}\,c_{\bm{k}}^g\,e^{2\pi\mathrm{i}\bm{k}\cdot\bm{x}}\\
\displaystyle
\widetilde{\widetilde{f}}=\sum_{\bm{l}\neq\bm{0}}\,\left(\prod_{i=1}^n\,\bm{u}_-^{(i)}\cdot\bm{l}\right)\,c_{\bm{l}}^f\,e^{2\pi\mathrm{i}\bm{l}\cdot\bm{x}}
\end{array}
\end{equation}
Finally, we observe that the product over the eigenvalues can be written as 
\begin{equation}
\label{prod2KS}
\left(\prod_{i=1}^n\,\rho_-^{(i)}\right)^r=e^{-r\sum_{i=1}^n\,\log\,\rho_+^{(i)}}=e^{-r  S_{\mathrm{K-S}} }
\end{equation}
where $S_{\mathrm{K-S}}$ is the Kolmogorov--Sinai entropy of the system of $n$ coupled maps. 

Therefore, we conclude that 
\begin{equation}
\label{Crfgasympt}
|C_r(f,g^\ast)|\sim e^{-r S_{\mathrm{K-S}}},
\end{equation}
as $r\to\infty,$ and this allows us to identify the mixing time of the system of $n$ coupled ACM maps with $1/S_{\mathrm{K-S}}:$
\begin{equation}
\label{mixingtimeACML}
\tau_\mathrm{mixing}(\mathrm{ACM\,lattice})=\frac{1}{S_{\mathrm{K-S}}}
\end{equation}
The steps of this 	computation can be straightfowardly generalized to the case of so--called ``$l-$fold mixing'', as was done for the single cat map in the previous section.

These bounds  hold even for the case of observables with square integrable first order partial derivatives. In the case that higher order derivatives are square integrable, we can use correspondingly higher powers of the differential operators $D_{u_-}^{(i)}$ and obtain, correspondingly, better bounds.

\section{Conclusions and outlook}\label{concl}
In this work we have studied,  in detail, the large time asymptotic behavior of the correlation functions of observables, that describe the mixing properties of a special class of automorphisms, those of  coupled Arnol'd cat maps, which are both symplectic and hyperbolic, on toroidal phase spaces and have computed the mixing time, in closed form. 
We have shown that the mixing time is given by  $1/S_\mathrm{K-S},$ where $S_\mathrm{K-S}$ is the Kolmogorov-Sinai entropy. 

In the literature numerical studies of nonlinear systems with many degrees of freedom have shown similar dependence of the mixing time on the Kolmogorov-Sinai entropy; in our case, however, what is of interest is that  we have an analytic result, which can be used to study the problem of relaxation of chaotic systems, subject to localized initial perturbations. 
There aren't many examples of field theories, with tuneable non-locality, for which the relaxation dynamics can be analytically controlled and can capture the salient properties of holographic systems; our work adds new classes to this list.

At the classical level it seems that there isn't any bound on how small the mixing time of the system can be; on the other hand, at the quantum level, studies of many-body quantum systems and the evolution of localized perturbations of black hole horizons seem to indicate that such a bound does, in fact, exist--the so-called ``scrambling time bound''~\cite{Sekino:2008he}--and is given by log $S_\mathrm{BH}$, where $S_\mathrm{BH}$ is the black hole entropy. 

The study of whether this bound is satisfied--or not--by the quantization of our system will be reported in future work. 

Our calculation  is, also,  relevant for the so-called Rokhlin conjecture, for which we clarify the conditions under which it may hold. 

Our results allow us to focus on the conditions for physical systems, that can minimize the mixing time, in the classical limit and set the stage for addressing how the mixing time of classical many-body systems is related to the scrambling time of the corresponding quantum systems. This is of relevance for understanding, on the one hand, transport properties of novel quantum materials, as well as the properties of quantum black holes~\cite{Hayden:2007cs,Sekino:2008he,maldacena2016bound}. To this end it is necessary to construct the corresponding unitary evolution operators for coupled Arnol'd cat maps, going beyond our previous work~\cite{Axenides:2013iwa}, for the case of the single Arnol'd cat map.

{\bf Acknowledgements:} This research was partially supported by the CNRS  International Emerging Actions program, “Chaotic behavior of closed quantum systems” under contract 318687.

 \bibliographystyle{utphys}
\bibliography{MAads2discrete}
\end{document}